\documentclass[aps,nopacs,nokeys,11pt,twoside,notitlepage,superscriptaddress]{revtex4-1}

\usepackage{graphicx,epic,eepic,epsfig,amsmath,amsfonts,latexsym,amssymb,verbatim,color}

\def\squareforqed{\hbox{\rlap{$\sqcap$}$\sqcup$}}
\def\qed{\ifmmode\squareforqed\else{\unskip\nobreak\hfil
\penalty50\hskip1em\null\nobreak\hfil\squareforqed
\parfillskip=0pt\finalhyphendemerits=0\endgraf}\fi}
\def\endenv{\ifmmode\;\else{\unskip\nobreak\hfil
\penalty50\hskip1em\null\nobreak\hfil\;
\parfillskip=0pt\finalhyphendemerits=0\endgraf}\fi}

\newtheorem{theorem}{Theorem}

\newtheorem{lemma}[theorem]{Lemma}

\newtheorem{proposition}[theorem]{Proposition}

\newenvironment{proof}[1][Proof]{\noindent\textbf{Proof.} }{\hfill\qed}
\newenvironment{proofof}[1][Proof]{\noindent\textbf{Proof~#1.} }{\hfill\qed}
\newcommand{\ket}[1]{|#1\rangle}

\mathchardef\ordinarycolon\mathcode`\:
\mathcode`\:=\string"8000
\def\vcentcolon{\mathrel{\mathop\ordinarycolon}}
\begingroup \catcode`\:=\active
  \lowercase{\endgroup
  \let :\vcentcolon
  }

\newcommand{\nc}{\newcommand}
\nc{\rnc}{\renewcommand}
\nc{\beq}{\begin{equation}}
\nc{\eeq}{{\end{equation}}}
\nc{\beqa}{\begin{eqnarray}}
\nc{\eeqa}{\end{eqnarray}}
\nc{\lbar}[1]{\overline{#1}}
\nc{\ketbra}[2]{|#1\rangle\!\langle#2|}
\nc{\proj}[1]{| #1\rangle\!\langle #1 |}
\nc{\avg}[1]{\langle#1\rangle}
\rnc{\max}{\operatorname{max}}
\nc{\Rank}{\operatorname{Rank}}
\nc{\smfrac}[2]{\mbox{$\frac{#1}{#2}$}}
\nc{\tr}{\operatorname{Tr}}
\nc{\ox}{\otimes}
\nc{\dg}{\dagger}
\nc{\dn}{\downarrow}
\nc{\cA}{\mathcal{A}}
\nc{\cB}{\mathcal{B}}
\nc{\cC}{\mathcal{C}}
\nc{\cD}{\mathcal{D}}
\nc{\cE}{\mathcal{E}}
\nc{\cF}{\mathcal{F}}
\nc{\cG}{\mathcal{G}}
\nc{\cH}{\mathcal{H}}
\nc{\cI}{\mathcal{I}}
\nc{\cJ}{\mathcal{J}}
\nc{\cK}{\mathcal{K}}
\nc{\cL}{\mathcal{L}}
\nc{\cM}{\mathcal{M}}
\nc{\cN}{\mathcal{N}}
\nc{\cO}{\mathcal{O}}
\nc{\cP}{\mathcal{P}}
\nc{\cR}{\mathcal{R}}
\nc{\cS}{\mathcal{S}}
\nc{\cT}{\mathcal{T}}
\nc{\cX}{\mathcal{X}}
\nc{\cZ}{\mathcal{Z}}
\nc{\supp}{{\operatorname{supp}}}
\nc{\qsupp}{{\operatorname{qsupp}}}
\nc{\var}{\operatorname{var}}
\nc{\rar}{\rightarrow}
\nc{\lrar}{\longrightarrow}
\nc{\polylog}{\operatorname{polylog}}
\nc{\1}{{\openone}}
\nc{\id}{{\operatorname{id}}}

\nc{\RR}{{{\mathbb R}}}
\nc{\CC}{{{\mathbb C}}}
\nc{\FF}{{{\mathbb F}}}
\nc{\NN}{{{\mathbb N}}}
\nc{\ZZ}{{{\mathbb Z}}}
\nc{\PP}{{{\mathbb P}}}
\nc{\QQ}{{{\mathbb Q}}}
\nc{\UU}{{{\mathbb U}}}
\nc{\EE}{{{\mathbb E}}}
\nc{\rG}{{{\mathrm G}}}

\nc{\be}{\begin{equation}}
\nc{\ee}{{\end{equation}}}
\nc{\bea}{\begin{eqnarray}}
\nc{\eea}{\end{eqnarray}}
\nc{\<}{\langle}
\nc{\Hom}[2]{\mbox{Hom}(\CC^{#1},\CC^{#2})}
\nc{\rU}{\mbox{U}}

\def\sep{\mathinner{\mathrm{SEP}}}

\nc{\LO}{\mathsf{LO}}
\nc{\ONELOCC}{\mathsf{1\text{-}LOCC}}
\nc{\LOCC}{{\mathsf{LOCC}}}
\nc{\SEP}{{\mathsf{SEP}}}
\nc{\PPT}{{\mathsf{PPT}}}

\nc{\ERLO}{{E_{r,\mathsf{LO}}}}
\nc{\ERONELOCC}{{E_{r,\mathsf{1\text{-}LOCC}}}}
\nc{\ERLOCC}{{E_{r,\mathsf{LOCC}}}}
\nc{\ERSEP}{{E_{r,\mathsf{SEP}}}}
\nc{\ERPPT}{{E_{r,\mathsf{PPT}}}}
\nc{\ERLOCCinfty}{{E^{\infty}_{r,\mathsf{LOCC}}}}

\begin{document}

\title{Relative entropy and squashed entanglement}

\author{Ke Li}
    \email{carl.ke.lee@gmail.com}
    \affiliation{Centre for Quantum Technologies, National University of Singapore, 3 Science Drive 2, Singapore 117543}
\author{Andreas Winter}
	\email{der.winter@gmail.com}
	\affiliation{ICREA -- Instituci\'{o} Catalana de Recerca i Estudis Avan\c{c}ats, Pg.~Lluis Companys 23, ES-08010 Barcelona, Spain}
    \affiliation{F\'{\i}sica Te\`{o}rica: Informaci\'{o} i Fenomens Qu\`{a}ntics, Universitat Aut\`{o}noma de Barcelona, ES-08193 Bellaterra (Barcelona), Spain}
	\affiliation{Department of Mathematics, University of Bristol, Bristol BS8 1TW, U.K.}
	\affiliation{Centre for Quantum Technologies, National University of Singapore, 3 Science Drive 2, Singapore 117543}
	
\date{15 October 2012}


\begin{abstract}
We are interested in the properties and relations of entanglement measures. Especially, we
focus on the squashed entanglement and relative entropy of entanglement, as well as their
analogues and variants.

Our first result is a monogamy-like inequality involving the relative entropy of entanglement
and its one-way LOCC variant. The proof is accomplished by exploring the properties of
relative entropy in the context of hypothesis testing via one-way LOCC operations, and by
making use of an argument resembling that by Piani on the faithfulness of regularized
relative entropy of entanglement.

Following this, we obtain a commensurate and faithful lower bound for
squashed entanglement, in the form of one-way LOCC relative entropy of entanglement.
This gives a strengthening to the strong subadditivity of von Neumann entropy. Our result improves the
trace-distance-type bound derived in [Comm. Math. Phys., 306:805-830, 2011], where
faithfulness of squashed entanglement was first proved. Applying Pinsker's inequality, we
are able to recover the trace-distance-type bound, even with slightly better constant factor.
However, the main improvement is that our new lower bound can be much larger than the old one
and it is almost a genuine entanglement measure.

We evaluate exactly the various relative entropy of entanglement under restricted measurement
classes, for maximally entangled states. Then, by proving asymptotic continuity, we extend the
exact evaluation to their regularized versions for all pure states. Finally, we consider
comparisons and separations between some important entanglement measures and obtain several
new results on these, too.
\end{abstract}

\maketitle

\section{Squashed entanglement and other entanglement measures}
\label{sec:intro}
As an important concept in quantum mechanics, entanglement plays a central role in
quantum information processing. It is the resource responsible for the
quantum computational speed-up, quantum communication, quantum cryptography and
so on.
Mathematically, quantum entanglement is the the most outstanding non-classical feature of
compound states that cannot be decomposed as statistical mixtures of product states over
subsystems, and has been found to possess a very rich structure. There exist many entanglement
measures, defined under various motivations and each characterizing some of its
features. The properties and relations of these entanglement measures are very much desirable
for our understanding of entanglement. Despite considerable achievements, a lot of issues
still remain unclear, even in the bipartite case~\cite{Horodeckis07}.

Among all the existing entanglement measures, squashed
entanglement~\cite{Tucci99, Tucci02, Christandl-Winter03} is a particularly
interesting one, with many desirable properties.
In analogy to the classical intrinsic information~\cite{Maurer-Wolf99},
squashed entanglement of a bipartite quantum state $\rho_{AB}$ is
defined as
\begin{equation}
  \label{eq:Esq-definition}
  E_{sq}(\rho_{AB}):=\inf\left\{\frac{1}{2}I(A;B|E)_\rho: \rho_{ABE}
                                            \text{ is an extension of } \rho_{AB}\right\},
\end{equation}
where $I(A;B|E)_\rho$ is the quantum conditional mutual information of $\rho_{ABE}$,
\begin{equation}
  I(A;B|E)_\rho:=S(\rho_{AE})+S(\rho_{BE})-S(\rho_{ABE})-S(\rho_{E})
\end{equation}
with the von Neumann entropy $S(\rho):=-\tr \rho\log\rho$. Squashed entanglement satisfies
most of the properties that are desired or useful for an entanglement measure. For example,
it is monotone under LOCC operations, convex and asymptotically continuous as a function of
quantum states, monogamous among one party and other parties, additive on tensor products
and superadditive in general~\cite{Christandl-Winter03, Alicki-Fannes04, Koashi-Winter04}.
Moreover, squashed entanglement admits an operational interpretation:
it is the minimum rate of qubits transmission at which a quantum state can be
redistributed among two parties when arbitrary (quantum) side information is
permitted~\cite{Devetak-Yard07, Yard-Devetak08, Ye-Bai-Wang08, Oppenheim08}.

Quantum relative entropy, given by
\[
  D(\rho\|\sigma) = \begin{cases}
                      \tr(\rho(\log\rho-\log\sigma)) & \text{ if }\supp(\rho)\subseteq\supp(\sigma), \\
                      +\infty                        & \text{ otherwise,}
                    \end{cases}
\]
measures the distinguishability of two states $\rho$ and $\sigma$ in the context of asymmetric
hypothesis testing~\cite{Hiai-Petz91, Ogawa-Nagaoka00}.
Yet it has found important applications in other aspects of quantum information theory:
The relative entropy of entanglement~\cite{VPRK97, Vedral-Plenio98} is another entanglement
measure that is of fundamental importance. For composite system $A \ox B$,
let $\sep(A:B)$ be the set of all separable states, i.e., the states of
the form $\sigma_{AB}=\sum_ip_i\sigma^A_i\ox\sigma^B_i$.
Relative entropy of entanglement,
\begin{equation}
  E_r(\rho_{AB}):=\min_{\sigma_{AB}\in\sep} D(\rho\|\sigma),
\end{equation}
quantifies the amount of entanglement of a state $\rho_{AB}$, by its relative entropy
``distance'' to the nearest separable state. Since relative entropy of entanglement is
strictly subadditive~\cite{Voll-Werner01}, it is more meaningful in many circumstances
to use its regularization,
\[
  E^\infty_r(\rho_{AB}):=\lim\limits_{n\rightarrow\infty}\frac{1}{n}E_r(\rho_{AB}^{\ox n}).
\]
Brand\~{a}o and Plenio have provided operational interpretations
to $E^\infty_r$: it quantifies the optimal rate of transformation between a quantum
state and maximally entangled states under non-entangling
operations~\cite{Brandao-Plenio-1, Brandao-Plenio-2}, and it is also the best error
exponent in quantum hypothesis testing where one of the hypothesis is many copies of the
state and the other one is the set of separable states~\cite{Brandao-Plenio-3}.

For each positive operator-valued measurement (POVM) $\{M_i\}_i$, it can be alternatively
identified with a measurement operation $\cM$, which is a completely positive map from
density matrices to probability vectors,
\[
  \cM(\omega)=\sum_i\proj{i}\tr(\omega M_i).
\]
On composite system $AB$, we define some restricted classes of measurements $\LO$,
$\ONELOCC$, $\LOCC$, $\SEP$ and $\PPT$. Here $\LO$, $\ONELOCC$ and $\LOCC$ are the sets
of measurements that can be implemented by means of local operations, local operations
and one-way classical communication, local operations and arbitrary two-way classical
communication, respectively; $\SEP$ and $\PPT$ are the classes of measurements whose POVM
elements are separable or positive-partial-transpose, respectively. Without loss of generality,
we assume that the one-way classical communication in $\ONELOCC$ is always from $A$
to $B$.

We see from the definition that squashed entanglement is always non-negative, due to the
strong subadditivity of von Neumann entropy, which states that the quantum conditional
mutual information can not be negative~\cite{Lieb-Ruskai73}. However, until very recently
proven in~\cite{BCY10}, the faithfulness of squashed entanglement, meaning that a bipartite
quantum state has non-vanishing squashed entanglement if and only if it is entangled, had
been a long-standing open question. Note that the infimum in the definition of
Eq.~(\ref{eq:Esq-definition}) cannot be replaced by minimum, because no bound on the
dimension of the system $E$ is known. As a result, the evaluation of squashed entanglement
becomes very difficult.

The main result of the proof in~\cite{BCY10} is the following inequality:
\begin{equation}
  \label{eq:BCY-bound}
  E_{sq}(\rho_{AB})\geq \frac{1}{16\ln 2}\min_{\sigma_{AB}\in\sep}\left\|\rho_{AB}-\sigma_{AB}
                                                                   \right\|^2_{\ONELOCC},
\end{equation}
where
\[
  \left\|\rho_{AB}-\sigma_{AB}\right\|_{\ONELOCC}
    := \sup_{\cM\in\ONELOCC} \left\|\cM(\rho_{AB})-\cM(\sigma_{AB})\right\|
\]
defines a metric (in fact, a norm) on density operators~\cite{MWW08}.

\medskip
The rest of the paper is structured as follows.
In Section~\ref{sec:main results} we state our main results. Then, after
considering quantum hypothesis testing under one-way LOCC measurements and obtaining a key
technical lemma in Section~\ref{sec:hypo}, we prove these results in
Sections~\ref{sec:proof-of-monogamy-newbound}, \ref{sec:proof-asycontinuity}
and \ref{sec:MES-Pure-states}, respectively. In Section~\ref{sec:comparison},
we deal with the comparisons and separations between entanglement measures and end
the paper with a few open questions.

\section{Main results}
\label{sec:main results}
Before presenting the results, we introduce the variants of relative entropy of
entanglement, which will be involved intensively later.
Piani defined the relative entropy of entanglement with respect to the set of states
$\rG$ and the restricted class of measurements $\mathsf{M}$~\cite{Piani09}, as
\begin{equation}
  \label{eq:ERGM-definition}
  E^{(\rG)}_{r,\mathsf{M}}(\rho)
     := \inf_{\sigma\in\rG} \sup_{\cM\in\mathsf{M}} D\bigl( \cM(\rho) \| \cM(\sigma) \bigr).
\end{equation}
Using this entanglement measure, he proved that $E_r^\infty$ is faithful, i.e.,
$E_r^\infty(\rho_{AB})>0$ if and only if $\rho_{AB}$ is entangled (same result was derived
in~\cite{Brandao-Plenio-3} independently).

In this paper, $\rG$ is usually the set of separable states $\sep$. Therefore, we abbreviate
$E^{(\sep)}_{r,\mathsf{M}}$ to $E_{r,\mathsf{M}}$ for simplicity.

\medskip\noindent
{\bf Monogamy relation for relative entropy of entanglement.}
One of the most fundamental properties of entanglement is monogamy: the more a quantum
system is entangled with another, then the less it is entangled with the others. For
any entanglement measure $f$, one would expect a quantitative characterization of
monogamy of the form
\[f(\rho_{1:23})\geq f(\rho_{1:2})+ f(\rho_{1:3}).
\]
Although this is really the case for squashed entanglement~\cite{Koashi-Winter04},
relative entropy of entanglement -- along with many other entanglement measures --
does not satisfy such a strong relation, with the antisymmetric state being a
counterexample~\cite{CSW09-1, CSW09-2}.

Here, we propose and prove a properly weakened monogamy inequality for relative
entropy of entanglement, by invoking its one-way LOCC variant.

\begin{theorem}
  \label{theorem:Er-monogamy}
  For every tripartite quantum state $\rho_{ABE}$, we have
  \begin{equation}
  \label{eq:Er-monogamy-a}
    E_r(\rho_{B:AE})\geq E_{r,\ONELOCC}(\rho_{AB})+E_r^\infty(\rho_{BE}),
  \end{equation}
  and
  \begin{equation}
  \label{eq:Er-monogamy-b}
    E_r^\infty(\rho_{B:AE})\geq E_{r,\ONELOCC}^\infty(\rho_{AB})+E_r^\infty(\rho_{BE}).
  \end{equation}
\end{theorem}

Eq.~(\ref{eq:Er-monogamy-b}) is obtained from Eq.~(\ref{eq:Er-monogamy-a}) by regularizing both
sizes, and it becomes stronger due to the subadditivity of $E_r$ and superadditivity of
$E_{r,\ONELOCC}$~\cite{Voll-Werner01, Piani09}.

It is worth mentioning that Eq.~(\ref{eq:Er-monogamy-a}) and Eq.~(\ref{eq:Er-monogamy-b}) are
in the form similar to Piani's superadditivity-like relation
\[
  E_r(\rho_{A_1A_2:B_1B_2}) \geq E_{r,\mathsf{M}}(\rho_{A_1B_1}) + E_r(\rho_{A_2B_2}),
\]
with $\mathsf{M}$ be $\LOCC$ or $\SEP$. The difference is that in our result,
there is only one single system $B$ on the left side, while it appears twice on the right side.
As a result, the price we have to pay is degrading the measurement class to $\ONELOCC$
and imposing a regularization in the two terms of the right side, respectively
(see Eq.~(\ref{eq:Er-monogamy-a})). One the other hand, our proof needs new technique
(Lemma~\ref{lemma:nondemolition-measurement} in the next section), which is derived in the
context of quantum hypothesis testing under restricted measurement class $\ONELOCC$.

\medskip\noindent
{\bf Commensurate lower bound for squashed entanglement.}
We provide in this paper a commensurate
and faithful lower bound for squashed entanglement. Instead of the one-way LOCC trace distance as
in Eq.~(\ref{eq:BCY-bound}), our result is in the form of one-way LOCC relative entropy of
entanglement, which is more natural and stronger.
\begin{theorem}
  \label{theorem:squshed-new-bound}
  For any quantum state $\rho_{AB}$, we have
  \begin{equation}
    \label{eq:new-bound}
    E_{sq}(\rho_{AB})\geq \frac{1}{2}E_{r,\ONELOCC}^\infty(\rho_{AB})
                                                  \geq \frac{1}{2}E_{r,\ONELOCC}(\rho_{AB}).
  \end{equation}
\end{theorem}

The core inequality for von Neumann entropy, strong subadditivity, states that for any
tripartite state $\rho_{ABE}$,
\[ I(A;B|E)_\rho \geq 0. \]
Recalling the definition of squashed entanglement, Theorem~\ref{theorem:squshed-new-bound}
implies
\[ I(A;B|E)_\rho \geq E_{r,\ONELOCC}(\rho_{AB}), \]
and hence strengthens the strong subadditivity inequality by relating it to a distance-like
entanglement measure on two of the subsystems.

To see how our result of Theorem~\ref{theorem:squshed-new-bound} improves the lower bound
proven in~\cite{BCY10}, we explain in more detail as follows. On the one hand, applying
Pinsker's inequality~\cite{Fuchs-Graaf99}, we are able to recover the trace-distance bound
of Eq.~(\ref{eq:BCY-bound}), even with a slightly better constant factor:
\begin{equation*}
    E_{sq}(\rho_{AB})\geq  \frac{1}{4\ln 2}\min_{\sigma_{AB}\in\sep}\left\|\rho_{AB}-
                                                       \sigma_{AB}\right\|^2_{\ONELOCC}.
\end{equation*}
On the other hand, while the trace-distance bound can be at most $O(1)$, our new
bound (\ref{eq:new-bound}) can be very large. Indeed, $E_{r,\ONELOCC}$ is asymptotically
normalized, in the sense of Proposition~\ref{proposition:EPR-evaluation}.

\medskip\noindent
{\bf Asymptotic continuity.}
To quantify the resources in quantum protocols in a
physically robust way, entanglement measures are expected to be asymptotically continuous.
 Piani's paper~\cite{Piani09} contains the proofs of several properties of
$E^{(\rG)}_{r,\mathsf{M}}$ for certain combination of $\rG$ and $\mathsf{M}$.
Now we also show asymptotic continuity under very general conditions.

We say that a set $\mathrm{S}$ is star-shaped with respect to some $x_0\in\mathrm{S}$,
if $px+(1-p)x_0 \in \mathrm{S}$ for all $x\in\mathrm{S}$ and $0 \leq p \leq 1$.

\begin{proposition}
  \label{proposition:asycontinuity}
  Let $\mathsf{M}$ be any set of measurements, and $\rG$ be a set of states on a quantum
  system with Hilbert space dimension $k$, containing the maximally mixed state $\tau$
  and such that in fact $\rG$ is star-shaped with respect to $\tau$.
  Let $\rho,\rho'$ be two states of the quantum system with
  $\| \rho-\rho' \|_{\mathsf{M}} \leq \epsilon \leq \frac{1}{e}$.  Then,
  \[
    \bigl| E^{(\rG)}_{r,\mathsf{M}}(\rho) - E^{(\rG)}_{r,\mathsf{M}}(\rho') \bigr|
                                  \leq 2\epsilon \log \frac{6k}{\epsilon}.
  \]
\end{proposition}

\medskip\noindent
{\bf Evaluation on maximally entangled states and pure states.}
The entanglement measure
$E_{r,\mathsf{M}}$ is difficult to calculate due to the two optimizations in its definition.
Here we conduct the first exact evaluation on maximally entangled states, with $\mathsf{M}$
be any of $\{\LO, \ONELOCC, \LOCC, \SEP, \PPT \}$. The basic idea is to make use of the
symmetry of $\frac{1}{\sqrt{d}}\sum_{i=1}^d\ket{ii}$, namely, invariance under unitary
operation $U\otimes\overline{U}$. Then, with the help of asymptotic continuity of
Proposition~\ref{proposition:asycontinuity}, we further obtain their regularized versions
on general pure states.

At first glance, the restricted class of measurements $\mathsf{M}$ may make $E_{r,\mathsf{M}}$
much smaller than the normal relative entropy of entanglement. However, in our case we find
that they are almost the same when the local dimension is very large.
\begin{proposition}
\label{proposition:EPR-evaluation}
For the rank-$d$ maximally entangled state $\Phi_d$,
\begin{equation}
  \label{eq:ebit}
  \ERLO(\Phi_d) = \ERONELOCC(\Phi_d) = \ERLOCC(\Phi_d) = \ERSEP(\Phi_d)= \ERPPT(\Phi_d)
                = \log (d+1) - 1.
\end{equation}
As a corollary, this implies that for pure state $\psi_{AB}$, the regularized
versions are equal to the entropic pure state entanglement:
\begin{equation}
  \label{eq:purestates}
  E_{r,\LO}^\infty(\psi_{AB}) =
   E_{r,\ONELOCC}^\infty(\psi_{AB}) =
            E_{r,\LOCC}^\infty(\psi_{AB}) =
                  E_{r,\SEP}^\infty(\psi_{AB}) =
                       E_{r,\PPT}^\infty(\psi_{AB}) = S(\tr_B \psi).
\end{equation}
\end{proposition}

\section{Quantum hypothesis testing under \protect\\ one-way LOCC operations with limited disturbance}
\label{sec:hypo}
In quantum hypothesis testing, we are given many copies of an information source, which
is statistically described by state $\rho$ (the null hypothesis) or $\sigma$ (the
alternative hypothesis). The task is to decide which state the source is really
in. This is achieved by doing a two outcome measurement
$\{L_n,\1-L_n\}$
on $n$ realizations of the source. We define two types of errors.
Type \uppercase\expandafter{\romannumeral1}
error is the probability that we falsely conclude that the state is $\sigma$ while it
is actually $\rho$, given by $\alpha_n(L_n):=\tr \rho^{\ox n}(\1-L_n)$;
type \uppercase\expandafter{\romannumeral2} error instead is the probability that we
mistake $\sigma$ for $\rho$, given by $\beta_n(L_n):=\tr \sigma^{\ox n}L_n$. In an
asymmetric situation, we want to minimize the type \uppercase\expandafter{\romannumeral2}
error while only simply requiring that the type \uppercase\expandafter{\romannumeral1}
error converges to $0$. The quantum Stein's lemma states that the maximal error exponent
of type \uppercase\expandafter{\romannumeral2} is the relative entropy
$D(\rho\|\sigma)$~\cite{Hiai-Petz91, Ogawa-Nagaoka00}: On the one hand, there exists
a test $\{L_n, \1-L_n\}$ satisfying
\[
  \alpha_n(L_n) \rightarrow 0 \quad\text{and}\quad
  -\frac{1}{n}\log\beta_n(L_n) \rightarrow D(\rho\|\sigma).
\]
On the other hand, if a test $\{L_n, \1-L_n\}$ is such that
\[
  \liminf_{n\rar\infty} -\frac{1}{n}\log\beta_n(L_n) > D(\rho\|\sigma),
\]
then $\alpha_n(L_n) \rightarrow 1$. This also applies to the classical
setting, if we replace quantum states $\rho$ and $\sigma$ by classical probability
distributions and the quantum measurement by a classical decision function~\cite{Cover-Thomas91}.

When $\rho$ and $\sigma$ are compound quantum states, it is natural to put locality
constraints on the measurements $\{L_n, \1-L_n\}$. In this case, the problem of
quantum hypothesis testing becomes much more difficult, and solutions are
known only in some specific situations~\cite{Matthews-Winter07, Owari-Hayashi11}.
Here, we focus on the family of measurements which are implementable by means of local
operations and one-way classical communication (one-way LOCC).
Our goal is not to derive a single-letter formula for the optimal error exponent; instead,
we are interested in how the disturbance on the quantum states induced by the measurement
is limited, when certain error exponent of type \uppercase\expandafter{\romannumeral2}
is achieved.

Let the null hypothesis and alternative hypothesis be $\rho_{ABE}^{\ox n}$ and
$\sigma_{ABE}^{\ox n}$, respectively. Let the allowed operations be restricted to
one-way LOCC which is performed on systems $A^n$ and $B^n$, with classical communication
from Alice's side ($A^n$) to Bob's side ($B^n$). On the one hand, it is easy to see that, for any one-way LOCC
measurement $\cM^{AB\rar X}$, $D\left(\cM(\rho)\|\cM(\sigma)\right)$ is an achievable
error exponent of type \uppercase\expandafter{\romannumeral2}. This is because, after
doing the measurement $\cM$ on each copy of the quantum states, the two states $\rho^{\ox n}$
and $\sigma^{\ox n}$ are replaced by classical probability distributions $(\cM(\rho))^{\ox n}$
and $(\cM(\sigma))^{\ox n}$. Then applying the Stein's lemma in the classical setting,
we know that there exists a classical decision rule which can achieve the above-mentioned
error figure. Hence, the corresponding quantum measurement $\{L_n, \1-L_n\}$ can be
constructed from $\cM^{\ox n}$ and this decision rule.
On the other hand, when the two
kinds of errors are sufficiently small, the one-way LOCC test $\{L_n, \1-L_n\}$ can be
performed in such a way that the reduced states on system $B^nE^n$, $\rho_{BE}^{\ox n}$
and $\sigma_{BE}^{\ox n}$, are kept almost undisturbed. This is a consequence of the
``gentle measurement lemma''~\cite{Winter98}. Note that, generally speaking, the full
states $\rho_{ABE}^{\ox n}$ and $\sigma_{ABE}^{\ox n}$ will be inevitably disturbed
significantly by the measurement,
because in the one-way  LOCC procedure, Bob's choice of measurement
is based on the outcome of Alice's measurement, and the extracting of such classical
information generically has to damage the states at Alice's side.
%

\begin{lemma}
  \label{lemma:nondemolition-measurement}
  For any two states $\rho_{ABE}$ and $\sigma_{ABE}$, and any one-way LOCC measurement
  $\cM^{AB\rar X}$ acting on system $AB$, with classical communication from $A$ to $B$,
  there exists a sequence of quantum instruments $\cT_n^{A^nB^n\rightarrow XB^n}$, which
  are implementable via local operations and one-way
  classical communication from $A^n$ to $B^n$, such that
  \begin{align}
    \label{eq:D-conserve}
    \lim_{n\rightarrow\infty} \frac{1}{n}
       D\bigl( \cT_n^c(\rho_{AB}^{\otimes n}) \| \cT_n^c(\sigma_{AB}^{\otimes n}) \bigr)
                                               &= D\bigl( \cM(\rho_{AB})\|\cM(\sigma_{AB}) \bigr), \\
    \label{eq:state-undisturb}
    \lim_{n\rightarrow\infty}
      \left\| \cT_n^q \otimes \id^{E^n}(\rho_{ABE}^{\otimes n})
             - \rho_{BE}^{\otimes n}\right\|_1 &= 0,
  \end{align}
  where $\cT_n^c:=\tr_{B^n}\circ \cT_n^{A^nB^n\rar XB^n}$, and
  $\cT_n^q:=\tr_X\circ\cT_n^{A^nB^n\rar XB^n}$.
\end{lemma}

\begin{proof}
Let the POVM elements of the measurement $\cM$ be
$\{R_k^A\ox S_{k,\ell}^B\}_{k,\ell}$, with $\sum_k R_k=\1^A$ and
$\sum_{\ell}S_{k,\ell}=\1^B$ for all $k$. Operationally, this means
that Alice does a measurement $\{R_k\}_k$ on system A, then she tells Bob the outcome
$k$, and according to what he receives, Bob does a measurement $\{S_{k,\ell}\}_{\ell}$
on the system B. For $\cM^{\otimes n}$ acting on $A^nB^n$, we
denote the measurement outcomes
$(k_1 k_2\ldots k_n, \ell_1 \ell_2 \ldots \ell_n) =: (k^n, \ell^n)$,
and the corresponding measurement elements
$\bigotimes_{i=1}^n (R_{k_i}^{A_i} \otimes S_{k_i,\ell_i}^{B_i}) =: R_{k^n}\ox S_{k^n, \ell^n}$.

For the problem of quantum hypothesis testing with the null hypothesis
$\rho_{ABE}^{\otimes n}$ and the alternative hypothesis $\sigma_{ABE}^{\otimes n}$,
and the permitted operations be one-way LOCC on parties $A^n$ and $B^n$, we consider
the protocol as follows. First, we apply the measurement $\cM$ to each copy of the
states $\rho$ and $\sigma$, resulting in classical probability distributions
$\cM^{\otimes n}(\rho^{\otimes n})$ and $\cM^{\otimes n}(\sigma^{\otimes n})$. Then,
we partition the set $\{(k^n, \ell^n)\}$ of all measurement
outcomes into two disjoint subsets $\cO_{n,\text{Null}}$ and $\cO_{n,\text{Alt}}$,
and make a classical decision: if the measurement outcome is in $\cO_{n,\text{Null}}$,
we infer that the state is $\rho^{\otimes n}$ (null hypothesis);
otherwise, it belongs to $\cO_{n,\text{Alt}}$ and we conclude that the state is
 $\sigma^{\otimes n}$ (alternative hypothesis). In such a protocol, the two
types of errors are
\begin{align}
  \alpha_n &= \sum_{(k^n, \ell^n)\in\cO_{n,\text{Alt}}}
                 \tr\rho_{AB}^{\otimes n} (R_{k^n}\ox S_{k^n,\ell^n}), \label{eq:null-infer-p} \\
  \beta_n  &= \sum_{(k^n, \ell^n)\in\cO_{n,\text{Null}}}
                 \tr\sigma_{AB}^{\otimes n} (R_{k^n}\ox S_{k^n,\ell^n}). \label{eq:alternative-infer-q}
\end{align}
By the classical Stein's lemma~\cite{Cover-Thomas91}, there exists a partition
$\{(k^n, \ell^n)\} = \cO_{n,\text{Null}} \stackrel{.}{\cup}\cO_{n,\text{Alt}}$
such that
\begin{align}
  \lim_{n\rar\infty}\alpha_n                &= 0, \label{eq:type-1-error} \\
  \lim_{n\rar\infty}-\frac{1}{n}\log\beta_n &= D(\cM(\rho)\|\cM(\sigma)), \label{eq:type-2-error}
\end{align}
which leads to
\begin{equation}
  \lim_{n\rightarrow\infty}\frac{1}{n}D(\{1-\alpha_n,\alpha_n\}\|\{\beta_n,1-\beta_n\})
                                          =D(\cM(\rho)\|\cM(\sigma)). \label{eq:stein-lemma}
\end{equation}

From now on, we fix such a partition of $\{(k^n,\ell^n)\}$ into
$\cO_{n,\text{Null}}$ and $\cO_{n,\text{Alt}}$. Let
\begin{equation}
  Q_{k^n,x}:=\sqrt{ \sum_{\ell^n:(k^n,\ell^n)\in\cO_{n,x}} S_{k^n,\ell^n}}, \label{eq:measurement-reform}
\end{equation}
where the index $x$ can be ``Null'' or ``Alt''. It is obvious that
$\{Q_{k^n,\text{Null}}, Q_{k^n,\text{Alt}}\}$ forms a complete set of Kraus operators,
i.e.~$Q_{k^n,\text{Null}}^\dagger Q_{k^n,\text{Null}}+Q_{k^n,\text{Alt}}^\dagger Q_{k^n,\text{Alt}}=\1^{B^n}$.
We are now ready for the definition of quantum instrument $\cT_n^{A^nB^n\rightarrow XB^n}$:
\begin{equation}
  \cT_n(\omega_{A^nB^n})
     :=\sum_{x=\text{Null},\text{Alt}}
         \proj{x}^X \otimes
         \sum_{k^n} \tr_{A^n} \bigl(\sqrt{R_{k^n}} \ox Q_{k^n,x}\bigr)
                                   \omega_{A^nB^n}
                              \bigl(\sqrt{R_{k^n}} \ox Q_{k^n,x}\bigr). \label{eq:instrument-definition}
\end{equation}

To complete the proof, we will demonstrate that $\cT_n$ satisfies all the requirements as
advertised. First, it is obvious that $\cT_n$ can be realized by means of one-way LOCC.
Alice does a measurement $\{R_{k^n}\}$ on the system $A^n$, then she communicates the outcome
$k^n$ to Bob; upon receiving $k^n$, Bob does a two-outcome measurement with Kraus operators
$\{Q_{k^n,\text{Null}},Q_{k^n,\text{Alt}}\}$ on the system $B^n$, at the same time he stores
the measurement results ``Null'' or ``Alt'' in the classical register $X$.

Secondly, we verify Eq.~(\ref{eq:D-conserve}). Clearly, we can write
\begin{equation}
  \cT_n\otimes\id^{E^n}(\rho_{ABE}^{\otimes n})
      = \sum_x \proj{x}^X \otimes \tilde{\rho}_{B^nE^n}^x, \label{eq:instrument-rho}
\end{equation}
with
\begin{equation}
  \label{eq:instrument-rho-ensemble}
  \tilde{\rho}_{B^nE^n}^x
      = \sum_{k^n} \tr_{A^n} \bigl(\sqrt{R_{k^n}} \ox Q_{k^n,x}\ox \1^{E^n}\bigr)
                                     \rho_{ABE}^{\otimes n}
                             \bigl(\sqrt{R_{k^n}} \ox Q_{k^n,x}\ox \1^{E^n}\bigr).
\end{equation}
Eqs.~(\ref{eq:null-infer-p}), (\ref{eq:measurement-reform}) and (\ref{eq:instrument-rho-ensemble})
together guarantee that
\begin{equation}
  \label{eq:prob-coincide}
  \tr\tilde{\rho}_{B^nE^n}^\text{Alt}  = \alpha_n \quad\text{and}\quad
  \tr\tilde{\rho}_{B^nE^n}^\text{Null} = 1-\alpha_n,
\end{equation}
which together with Eq.~(\ref{eq:instrument-rho}) results in
\begin{equation}
  \label{eq:prob-coincide-p}
  \cT_n^c(\rho_{AB}^{\otimes n})=(1-\alpha_n)\proj{\text{Null}}^X+\alpha_n\proj{\text{Alt}}^X.
\end{equation}
Similarly, from Eqs.~(\ref{eq:alternative-infer-q}), (\ref{eq:measurement-reform})
and (\ref{eq:instrument-definition}), we derive that
\begin{equation}
  \label{eq:prob-coincide-q}
  \cT_n^c(\sigma_{AB}^{\otimes n})=\beta_n\proj{\text{Null}}^X+(1-\beta_n)\proj{\text{Alt}}^X.
\end{equation}
So, Eqs.~(\ref{eq:stein-lemma}), (\ref{eq:prob-coincide-p}) and (\ref{eq:prob-coincide-q})
imply
\begin{equation*}
  \lim_{n\rightarrow\infty} \frac{1}{n}D(\cT_n^c(\rho_{AB}^{\otimes n})\|
                        \cT_n^c(\sigma_{AB}^{\otimes n}))=D(\cM(\rho_{AB})\|\cM(\sigma_{AB})),
\end{equation*}
which is exactly Eq.~(\ref{eq:D-conserve}).

Finally, we prove that $\cT_n$ satisfies Eq.~(\ref{eq:state-undisturb}). Making use of
Eqs.~(\ref{eq:instrument-rho}) and (\ref{eq:prob-coincide}), we have
\begin{equation}\begin{split}
  \label{eq:state-undisturb-1}
  \left\| \cT_n^q \otimes \id^{E^n}(\rho_{ABE}^{\otimes n})-\rho_{BE}^{\otimes n}\right\|_1
      &= \left\| \tilde{\rho}_{B^nE^n}^\text{Alt}+ \tilde{\rho}_{B^nE^n}^\text{Null}-\rho_{BE}
                                                                ^{\otimes n}\right\|_1     \\
      &\leq \alpha_n+\left\|  \tilde{\rho}_{B^nE^n}^\text{Null}-\rho_{BE}^{\otimes n}\right\|_1.
\end{split}\end{equation}
Paying attention to the definition of $\tilde{\rho}_{B^nE^n}^\text{Null}$, namely
Eq.~(\ref{eq:instrument-rho-ensemble}), we easily check that
\begin{equation}
  \label{eq:state-undisturb-2}
  \tilde{\rho}_{B^nE^n}^\text{Null}=\tr_{K^n}\sqrt{\Lambda}\tilde{\rho}_{K^nB^nE^n}\sqrt{\Lambda},
\end{equation}
where $\tilde{\rho}_{K^nB^nE^n}:=\tr_{A^n}\sum_{k^n}\proj{k^n}^{K^n}\otimes (\sqrt{R_{k^n}}
\rho_{ABE}^{\otimes n}\sqrt{R_{k^n}})$ is a normalized quantum state, and $\Lambda:=\sum_{k^n}
\proj{k^n}^{K^n}\otimes Q_{k^n,\text{Null}}^2$ is a POVM element satisfying $0\leq\Lambda\leq\1$.
As a result,
\begin{equation}\begin{split}
  \label{eq:state-undisturb-3}
  \left\|  \tilde{\rho}_{B^nE^n}^\text{Null}-\rho_{BE}^{\otimes n}\right\|_1
       &=    \left\|\tr_{K^n}\sqrt{\Lambda}\tilde{\rho}_{K^nB^nE^n}\sqrt{\Lambda}-
                                               \tr_{K^n}\tilde{\rho}_{K^nB^nE^n}\right\|_1 \\
       &\leq \left\|\sqrt{\Lambda}\tilde{\rho}_{K^nB^nE^n}\sqrt{\Lambda}-
                                                        \tilde{\rho}_{K^nB^nE^n}\right\|_1 \\
       &\leq 2\sqrt{1-\tr \tilde{\rho}_{K^nB^nE^n} \Lambda} \\
       &=    2\sqrt{\alpha_n},
\end{split}\end{equation}
where the first line is by Eq.~(\ref{eq:state-undisturb-2}) and the fact that
$\tr_{K^n}\tilde{\rho}_{K^nB^nE^n}=\rho_{BE}^{\otimes n}$, the second line is because of
the monotonicity of trace distance under partial trace, the third line makes use of the
gentle measurement lemma~\cite{Winter98}, and the last line follows from
Eqs.~(\ref{eq:prob-coincide}) and (\ref{eq:state-undisturb-2}). Eventually, inserting
Eq.~(\ref{eq:state-undisturb-3}) into Eq.~(\ref{eq:state-undisturb-1}), and invoking
Eq.~(\ref{eq:type-1-error}), we arrive at
\begin{equation*}
  \left\| \cT_n^q \otimes \id^{E^n}(\rho_{ABE}^{\otimes n})-
                                                      \rho_{BE}^{\otimes n}\right\|_1
  \leq \alpha_n+2\sqrt{\alpha_n} \rightarrow 0,
\end{equation*}
which is precisely Eq.~(\ref{eq:state-undisturb}).
\end{proof}

\section{Entanglement monogamy relation and \protect\\
         commensurate lower bound for squashed entanglement}
\label{sec:proof-of-monogamy-newbound}
\begin{proofof}[of Theorem~\ref{theorem:Er-monogamy}]
As discussed in Section~\ref{sec:main results}, it suffices to prove Eq.~(\ref{eq:Er-monogamy-a}).
Let $\sigma_{B:AE}$ be the nearest separable state to $\rho_{B:AE}$ with respect to
the measure of relative entropy. That is to say,
\begin{equation}
\label{eq:Er-monogamy-1}
E_r(\rho_{B:AE})=D(\rho_{ABE}\|\sigma_{ABE})
                =\frac{1}{n} D(\rho_{ABE}^{\otimes n}\|\sigma_{ABE}^{\otimes n}).
\end{equation}
Let $\cM^{AB\rar X}$ be an arbitrary one-way LOCC measurement. Applying
Lemma~\ref{lemma:nondemolition-measurement} to $\rho_{ABE}$, $\sigma_{ABE}$ and
$\cM^{AB\rar X}$, we know that there exists a sequence of quantum instruments
$\cT_n^{A^nB^n\rar XB^n}$, which are implementable via local operations and classical
communication from $A^n$ to $B^n$, such that
\begin{align}
  \lim_{n\rightarrow\infty} \frac{1}{n}D(\cT_n^c(\rho_{AB}^{\otimes n})\|\cT_n^c(\sigma_{AB}
               ^{\otimes n})) &= D(\cM(\rho_{AB})\|\cM(\sigma_{AB})),   \label{eq:Er-monogamy-2}\\
  \lim_{n\rightarrow\infty}\left\| \cT_n^q \otimes \id^{E^n}(\rho_{ABE}^{\otimes n})-
                             \rho_{BE}^{\otimes n}\right\|_1 &= 0, \label{eq:Er-monogamy-3}
\end{align}
where $\cT_n^c:=\tr_{B^n}\circ \cT_n^{A^nB^n\rar XB^n}$, and $\cT_n^q:=\tr_X\circ
\cT_n^{A^nB^n\rar XB^n}$. Write $\cT_n\ox \id^{E^n} (\rho_{ABE}^{\otimes n})=\sum_{i_n}p_{i_n}
\proj{i_n}^X\otimes\rho_{B^nE^n}^{i_n}$ and $\cT_n\ox \id^{E^n}(\sigma_{ABE}^{\otimes n})=
\sum_{i_n}q_{i_n}\proj{i_n}^X\otimes\sigma_{B^nE^n}^{i_n}$. It is easy to check that
\begin{equation}\begin{split}
  \label{eq:Er-monogamy-4}
  D(\cT_n\ox\id^{E^n}(\rho_{ABE}^{\otimes n})\|\cT_n\ox\id^{E^n}(\sigma_{ABE}^{\otimes n}))
    &= D(\cT_n^c(\rho_{AB}^{\otimes n})\|\cT_n^c(\sigma_{AB}^{\otimes n}))+
                                \sum_{i_n}p_{i_n}D(\rho_{B^nE^n}^{i_n}\|\sigma_{B^nE^n}^{i_n}) \\
    &\!\!\!\!\!\!\!\!\!\!\!\!\!\!\!\!\!\!\!\!\!\!\!\!\!\!\!\!\!\!\!\!\!\!\!
     \geq D(\cT_n^c(\rho_{AB}^{\otimes n})\|\cT_n^c(\sigma_{AB}^{\otimes n}))+D(\cT_n^q\ox\id^{E^n}
                              (\rho_{ABE}^{\otimes n})\|\sum_{i_n}p_{i_n}\sigma_{B^nE^n}^{i_n})\\
    &\!\!\!\!\!\!\!\!\!\!\!\!\!\!\!\!\!\!\!\!\!\!\!\!\!\!\!\!\!\!\!\!\!\!\!
     \geq D(\cT_n^c(\rho_{AB}^{\otimes n})\|\cT_n^c(\sigma_{AB}^{\otimes n}))
                                     +E_r\bigl(\cT_n^q\ox\id^{E^n}(\rho_{ABE}^{\otimes n})\bigr),
\end{split}\end{equation}
where the first line is by direct calculation, the second line follows from convexity
of quantum relative entropy, and for the last line, note that the state $\sum_{i_n}p_{i_n}
\sigma_{B^n:E^n}^{i_n}$ is still separable because of the LOCC feature of $\cT_n$. By the
Lindblad-Uhlmann theorem~\cite{Lindblad75,Uhlmann77}, quantum relative entropy is
monotonic under cptp quantum operations. So, combining Eqs.~(\ref{eq:Er-monogamy-1})
and (\ref{eq:Er-monogamy-4}) results in
\begin{equation}
\label{eq:Er-monogamy-5}
E_r(\rho_{B:AE})\geq \frac{1}{n}D(\cT_n^c(\rho_{AB}^{\otimes n})\|\cT_n^c(\sigma_{AB}
         ^{\otimes n}))+\frac{1}{n}E_r\bigl(\cT_n^q\ox\id^{E^n}(\rho_{ABE}^{\otimes n})\bigr).
\end{equation}
It was proven in~\cite{Donald-Horodecki99} that the relative entropy of entanglement satisfies
a strong continuity condition: for two states $\rho_1$ and $\rho_2$ on system $AB$ with
$\|\rho_1-\rho_2 \|_1 \leq \frac{1}{e}$, we have
\begin{equation}
\label{eq:Er-monogamy-6}
  |E_r(\rho_1)-E_r(\rho_2)| \leq 2(2+\log |A|+\log|B|) \|\rho_1-\rho_2 \|_1
                                 + 2\eta(\|\rho_1-\rho_2 \|_1),
\end{equation}
where $\eta(x)=-x\log x$. Now, letting $n\rar\infty$ in Eq.~(\ref{eq:Er-monogamy-5}),
and then making use of Eqs.~(\ref{eq:Er-monogamy-2}), (\ref{eq:Er-monogamy-3}) and
(\ref{eq:Er-monogamy-6}), we obtain
\begin{equation}\begin{split}
  \label{eq:Er-monogamy-7}
  E_r(\rho_{B:AE})&\geq \lim_{n\rar\infty}\frac{1}{n}D(\cT_n^c(\rho_{AB}^{\otimes n})
                     \|\cT_n^c(\sigma_{AB}^{\otimes n}))+\lim_{n\rar\infty}\frac{1}{n}
                     E_r\bigl(\cT_n^q\ox\id^{E^n}(\rho_{ABE}^{\otimes n})\bigr) \\
                  & = D(\cM(\rho_{AB})\|\cM(\sigma_{AB})) +  E_r^\infty(\rho_{BE}).
\end{split}\end{equation}
Since $\cM$ is arbitrary, it follows from Eq.~(\ref{eq:Er-monogamy-7}) that
\begin{equation}
  \label{eq:Er-monogamy-8}
  E_r(\rho_{B:AE}) \geq  \sup_{\cM\in\ONELOCC} D(\cM(\rho_{AB})\|\cM(\sigma_{AB}))
                  +E_r^\infty(\rho_{BE})\geq E_{r,\ONELOCC}(\rho_{AB})+E_r^\infty(\rho_{BE}),
\end{equation}
where the second inequality is by the definition of $E_{r,\ONELOCC}$, and we are done.
\end{proofof}

\medskip
\begin{proofof}[of Theorem~\ref{theorem:squshed-new-bound}]
It is shown in~\cite[Lemma 1]{BCY10} that
\begin{equation}
  \label{eq:proof-new-bound-1}
  I(A;B|E)_\rho \geq E_r^\infty(\rho_{B:AE})-E_r^\infty(\rho_{BE}).
\end{equation}
Eq.~(\ref{eq:Er-monogamy-b}) in Theorem~\ref{theorem:Er-monogamy}, together with
Eq.~(\ref{eq:proof-new-bound-1}), gives us
\begin{equation}
  \label{eq:proof-new-bound-2}
  I(A;B|E)_\rho \geq E_{r,\ONELOCC}^\infty(\rho_{AB}).
\end{equation}
Then, recalling the definition of squashed entanglement and by the superadditivity of
$\ERONELOCC$~\cite{Piani09}, we arrive at
\begin{equation*}
         E_{sq}(\rho_{AB})\geq \frac{1}{2}E_{r,\ONELOCC}^\infty(\rho_{AB})
    \geq \frac{1}{2}E_{r,\ONELOCC}(\rho_{AB}),
\end{equation*}
which concludes the proof.
\end{proofof}

\section{Asymptotic continuity}
\label{sec:proof-asycontinuity}
\begin{proofof}[of Proposition~\ref{proposition:asycontinuity}]
  For $0 \leq x \leq 1$, let $\rG_x := x\rG + (1-x)\tau$, so that $\rG_1 = \rG$ and
  $\rG_0 = \tau$. We follow very closely~\cite{Donald-Horodecki99}, and start by the
  observation that because of $\rG_x \subset \rG$ and the operator monotonicity of the
  $\log$ function,
  \begin{equation}
    \label{eq:D-Dx}
    E^{(\rG)}_{r,\mathsf{M}} \leq E^{(\rG_x)}_{r,\mathsf{M}} \leq E^{(\rG)}_{r,\mathsf{M}} - \log x.
  \end{equation}
  We will later see that $x=1-\epsilon$ is a good choice.
  However, it is clear already that if it is
  close to $1$, then we reduce our problem to proving
  asymptotic continuity for $\rG_x$, which has the
  property that all of its elements are of full rank. In fact, the smallest
  eigenvalue of a $\sigma \in \rG_x$ is $\geq \frac{1-x}{k}$.

  Now fix $\sigma \in \rG_x$ and $\cM \in \mathsf{M}$, and consider
  \[
    E^{(\sigma)}_{r,\{\cM\}}(\rho) = D\bigl( \cM(\rho) \| \cM(\sigma) \bigr)
                                 = \sum_i \tr \rho M_i \log \frac{\tr \rho M_i}{\tr \sigma M_i}.
  \]
  Since $0\leq M_i \leq \1$, we can write $M_i = 3k\lambda_i Q_i$ with
  operators $Q_i \geq 0$ s.t.~$\tr Q_i = \frac13$,
  and $\lambda_i \geq 0$, $\sum_i \lambda_i = 1$. Then,
  $\frac{1}{3} \geq \tr \sigma Q_i \geq \frac{1-x}{3k}$ for all $i$. We can also rewrite
  the above quantity as
  \[\begin{split}
    E^{(\sigma)}_{r,\{\cM\}}(\rho)
                   &= 3k \sum_i \lambda_i \tr \rho Q_i \log \frac{\tr \rho Q_i}{\tr \sigma Q_i} \\
                   &= - \sum_i \tr \rho M_i \log \tr \sigma Q_i +
                        3k \sum_i \lambda_i \tr \rho Q_i \log \tr \rho Q_i,
  \end{split}\]
  and we will treat the two latter sums separately; call them $\text{I}(\rho)$ and
  $\text{II}(\rho)$, respectively.
  For the first one,
  \[\begin{split}
    \bigl| \text{I}(\rho) - \text{I}(\rho') \bigr|
               &=    \left| \sum_i \tr(\rho-\rho')M_i \log \tr\sigma Q_i \right| \\
               &\leq \sum_i \log \frac{3k}{1-x} \bigl| \tr\rho M_i - \tr\rho' M_i \bigr| \\
               &=    \log \frac{3k}{1-x} \| \rho-\rho' \|_{\{\cM\}}
                \leq \epsilon \log\frac{3k}{\epsilon}.
  \end{split}\]
  For the second term, we use the function $\eta(t) = -t\log t$, which is
  concave, non-negative on the unit interval and has the elementary property that
  for all $s,t \geq 0$, $\eta(s+t) \leq \eta(s)+\eta(t)$. Furthermore, for
  $0\leq t \leq \frac{1}{e}$ is is monotonically increasing. Now,
  $\text{II}(\rho) = -3k \sum_i \lambda_i \eta(\tr\rho Q_i)$, and so
  \[\begin{split}
    \bigl| \text{II}(\rho) - \text{II}(\rho') \bigr|
               &\leq 3k \sum_i \lambda_i \bigl| \eta(\tr\rho Q_i) - \eta(\tr\rho' Q_i) \bigr| \\
               &\leq 3k \sum_i \lambda_i \eta\bigl( |\tr(\rho-\rho')Q_i| \bigr) \\
               &\leq 3k \eta\left( \sum_i \lambda_i |\tr(\rho-\rho')Q_i| \right) \\
               &=    3k \eta\left( \frac{1}{3k} \|\rho-\rho'\|_{\{\cM\}} \right) \\
               &\leq 3k \eta\left( \frac{\epsilon}{3k} \right)
                =    \epsilon \log\frac{3k}{\epsilon}.
  \end{split}\]
  where in the third line we have used the concavity of $\eta$.

  Putting these two observations together, we obtain (recall $\sigma \in \rG_x$,
  $x=1-\epsilon$)
  \begin{equation}
    \label{eq:fixed-state-and-POVM}
    \bigl| E^{(\sigma)}_{r,\{\cM\}}(\rho) - E^{(\sigma)}_{r,\{\cM\}}(\rho') \bigr|
                                                \leq 2\epsilon \log \frac{3k}{\epsilon}.
  \end{equation}
  From this, the rest of the argument is pretty standard, all we need to
  implement is the maximization over $\cM\in\mathsf{M}$ and the minimization
  over $\sigma\in\rG_x$.
  First, fix $\sigma\in\rG_x$; then,
  \[\begin{split}
    \bigl| E^{(\sigma)}_{r,\mathsf{M}}(\rho) - E^{(\sigma)}_{r,\mathsf{M}}(\rho') \bigr|
           &= \bigl| \sup_\cM E^{(\sigma)}_{r,\{\cM\}}(\rho) - \sup_{\cM'} E^{(\sigma)}_{r,\{\cM'\}}(\rho') \bigr| \\
           &\leq \sup_{\cM\in\mathsf{M}}
                  \bigl| E^{(\sigma)}_{r,\{\cM\}}(\rho) - E^{(\sigma)}_{r,\{\cM\}}(\rho') \bigr| \\
           &\leq 2\epsilon \log \frac{3k}{\epsilon},
  \end{split}\]
  invoking eq.~(\ref{eq:fixed-state-and-POVM}). Similarly,
  \[\begin{split}
    \bigl| E^{(\rG_x)}_{r,\mathsf{M}}(\rho) - E^{(\rG_x)}_{r,\mathsf{M}}(\rho') \bigr|
           &= \bigl| \inf_\sigma E^{(\sigma)}_{r,\mathsf{M}}(\rho)
                                       - \inf_{\sigma'} E^{(\sigma')}_{r,\mathsf{M}}(\rho') \bigr| \\
           &\leq \sup_{\sigma\in\rG_x}
                  \bigl| E^{(\sigma)}_{r,\mathsf{M}}(\rho) - E^{(\sigma)}_{r,\mathsf{M}}(\rho') \bigr| \\
           &\leq 2\epsilon \log \frac{3k}{\epsilon},
  \end{split}\]
  using the relation for fixed $\sigma$.
  From this and eq.~(\ref{eq:D-Dx}), using $-\log x = -\log(1-\epsilon) \leq 2\epsilon$,
  the proposition follows.
\end{proofof}

\section{Evaluation on maximally entangled states and pure states}
\label{sec:MES-Pure-states}
\begin{proofof}[of Proposition~\ref{proposition:EPR-evaluation}]
We show separately $\ERLO(\Phi_d) \geq \log(d+1)-1$ and $\ERPPT(\Phi_d) \leq \log(d+1)-1$,
which together complete the proof, since by definition,
$\ERLO \leq \ERONELOCC \leq \ERLOCC \leq \ERSEP \leq \ERPPT$.

For the former, we need to show that for each separable
state there exists an LO measurement such that
the relative entropy of the measurement outcomes is at least
$\log\frac{d+1}{2}$.
In fact, it suffices to employ the $U\ox\overline{U}$-twirl followed by
local measurements in the computational basis. Although this requires
shared randomness, it is easy to see that derandomization can be done due to the joint
convexity of relative entropy.
The twirl leaves $\Phi_d$ invariant and transforms the separable state into a
separable \emph{isotropic} state
\[
  \sigma = p\Phi_d + (1-p)\frac{1}{d^2}\1,
\]
where the separability is equivalent to $p\leq \frac{1}{d+1}$~\cite{Horodecki99}.
Now, the measurement of the maximally entangled state and of $\sigma$
yield distributions
\begin{align*}
  P(xy|\Phi_d) &= \frac{1}{d}\delta_{xy}, \\
  Q(xy|\sigma) &= \frac{p}{d}\delta_{xy} + \frac{1-p}{d^2}. \\
\end{align*}
From this it is straightforward to calculate the relative entropy
\[\begin{split}
  D(P\|Q) &= \sum_x \frac{1}{d}\log\frac{1/d}{p/d+(1-p)/d^2} \\
          &= -\log\left( p+\frac{1-p}{d} \right)             \geq \log\frac{d+1}{2},
\end{split}\]
and we are done.

For the second (upper) bound, we need to show that there is no better
measurement once we choose an appropriate separable state, which predictably
we set
\[
  \sigma = \frac{1}{d+1}\Phi_d + \frac{d}{d+1}\frac{1}{d^2}\1
         = \frac{1}{d}\Phi_d + \frac{d-1}{d}\frac{1}{d^2-1}(\1-\Phi_d).
\]
Now our entangled state and the separable candidate are isotropic,
This means that whatever PPT measurement we have, i.e.~with POVM
elements $M_i$ such that $M_i^\Gamma \geq 0$, the covariant POVM
$\bigl( {\rm d}U (U\ox\overline{U}) M_i (U\ox\overline{U})^\dagger \bigr)_{i,U}$
will achieve the same relative entropy. Note however that
the probabilities $\tr\rho(U\ox\overline{U}) M_i (U\ox\overline{U})^\dagger$
are independent of the unitary $U$ for isotropic $\rho \in \{ \Phi_d, \sigma \}$,
so we get the same relative entropy for the twirled POVM with operators
\[
  \widehat{M}_i = \int {\rm d}U (U\ox\overline{U}) M_i (U\ox\overline{U})^\dagger,
\]
which are all isotropic: $\widehat{M}_i = \alpha_i \Phi_d + \beta_i (\1-\Phi_d)$,
with $\alpha_i,\beta_i \geq 0$ and separately adding up to $1$. The PPT
condition is $\beta_i \geq \frac{1}{d+1}\alpha_i$ for all $i$.
Next, the maximum of the relative entropy will be attained on an extremal
measurement from this class, which restricts (w.l.o.g.) the number of outcomes
to two. The only nontrivial POVM with these properties is
composed of the two operators
\begin{align*}
  \widehat{M}_0 &= \Phi_d +           \frac{1}{d+1}(\1-\Phi_d), \\
  \widehat{M}_1 &= \phantom{\Phi_d +} \frac{d}{d+1}(\1-\Phi_d).
\end{align*}
For this measurement, the probabilities observed on $\Phi_d$ are $1$ and $0$,
respectively; for the above $\sigma$ they are $\frac{2}{d+1}$ and $\frac{d-1}{d+1}$,
yielding indeed a relative entropy of $\log\frac{d+1}{2}$.

\medskip
Now, we can conclude that $E_{r,\mathsf{M}}^\infty$, with $\mathsf{M}$ be any of
$\{ \LO, \ONELOCC, \LOCC, \SEP, \PPT \}$ coincides with the entropic entanglement
measure on pure states. This follows now easily from the asymptotic theory of pure
state entanglement and the asymptotic continuity. To be precise, let $\psi$ be a pure
state on $A \ox B$; then there is a sequence of $\epsilon_n \rightarrow 0$ and of
LO protocols(!) to convert $\psi^{\ox n}$ into $\rho^{(n)}$ with
$\left\| \Phi_2^{\ox n(E(\psi)-\epsilon_n)} - \rho^{(n)} \right\|_1 \leq \epsilon_n$.
By the monotonicity of $E_{r,\mathsf{M}}$ under local operations and for large enough $n$,
\[\begin{split}
  E_{r,\mathsf{M}}(\psi^{\ox n}) &\geq E_{r,\mathsf{M}}\left( \rho^{(n)} \right) \\
                        &\geq E_{r,\mathsf{M}}\left( \Phi_2^{\ox n(E(\psi)-\epsilon_n)} \right)
                                - 2\epsilon_n \log\frac{6\cdot2^{2nE(\psi)}}{\epsilon_n}\\
                        &\geq n E(\psi) - O(n)\epsilon_n - O(1).
\end{split}\]
Conversely, there are one-way LOCC protocols to convert
$\Phi_2^{\ox n(E(\psi)+\epsilon_n)}$ into $\omega^{(n)}$ with
$\left\| \psi^{\ox n} - \omega^{(n)} \right\|_1 \leq \epsilon_n$. Hence, for large
enough $n$,
\[\begin{split}
  E_{r,\mathsf{M}}(\psi^{\ox n}) &\leq E_{r,\mathsf{M}}\left( \omega^{(n)} \right)
                                + 2\epsilon_n \log\frac{6|A|^n |B|^n}{\epsilon_n} \\
                        &\leq E_{r,\mathsf{M}}\left( \Phi_2^{\ox n(E(\psi)+\epsilon_n)} \right)
                                + 2\epsilon_n \log\frac{6|A|^n |B|^n}{\epsilon_n} \\
                        &\leq n E(\psi) + O(n)\epsilon_n + O(1).
\end{split}\]

Together, we obtain, for $\mathsf{M} \in \{ \ONELOCC, \LOCC, \SEP, \PPT \}$,
\[
  \left| \frac{1}{n}E_{r,\mathsf{M}}(\psi^{\ox n}) - E(\psi) \right| \leq O(\epsilon_n) + O\left(\frac{1}{n}\right)
                                                           \rightarrow 0,
\]
as $n \rightarrow \infty$.
For $\LO$ the above reasoning does not apply because
we need one-way LOCC operations in the converse (dilution) part. However, as
$E_{r,\LO} \leq E_{r,\ONELOCC}$, the lower bound for the former and the upper
bound for the latter suffice.
\end{proofof}

\section{Comparisons between entanglement measures}
\label{sec:comparison}
In this section, we consider the relations between entanglement measures. Especially,
we are interested in two classes of them. The first class consists of squashed-like
measures. This includes the squashed entanglement $E_{sq}$ itself, the conditional entanglement
of mutual information $E_I(\rho_{AB}):=\frac{1}{2}\inf\{I(AA';BB')_\rho-I(A';B')_\rho\}$
with $\rho_{AA'BB'}$ being an extension of $\rho_{AB}$~\cite{YHW08}, and the c-squashed
entanglement $E_{sq,c}(\rho_{AB}):=\frac{1}{2}\inf\{I(A;B|E)_\rho\}$, where the infimum is
taken over all the extension state $\rho_{ABE}$ of the form $\sum_ip_i\rho^i_{AB}\ox
\proj{i}_E$~\cite{YHHHOS07}. It is known that these entanglement measures satisfy the
chain of inequalities~\cite{Christandl06,YHW07}
\begin{equation*}
E_d \leq K_d \leq E_{sq} \leq E_I \leq E_{sq,c}^\infty \leq E_c,
\end{equation*}
where $E_d$ is the distillable entanglement, $K_d$ is the distillable key and $E_c$
the entanglement cost.

The other class contains the relative entropy of entanglement $E_r$
and its relatives $E_{r,\leftrightarrow}(\rho_{AB}):=E_{r,\LOCC}(\rho_{AB})$ and
$E_{r,\rar}(\rho_{AB}):=\sup\{E_{r,\ONELOCC}(\Lambda(\rho_{AB})): \Lambda \text{ being LOCC}\}$.
Here $E_{r,\rar}$ is an ``update'' of $E_{r,\ONELOCC}$ such that it is LOCC monotone. Note that
in the definition of $E_{r,\rar}$, the supremum is taken over all LOCC operations, in contrast
to the smaller set of LOCC measurements. It is known that~\cite{Hors-Opp05}
\begin{equation*}
  E_d \leq K_d \leq E_r^\infty \leq E_c.
\end{equation*}
On the other hand, it is obvious from the definitions that
$E_{r,\rar}^\infty \leq E_{r,\leftrightarrow}^\infty \leq E_r^\infty$,
and we will show that $E_d \leq E_{r,\rar}^\infty$ later in
Proposition~\ref{proposition:relation-Errar-Ed}. Hence, we have also
\begin{equation*}
  E_d \leq E_{r,\rar}^\infty \leq E_{r,\leftrightarrow}^\infty \leq E_r^\infty \leq E_c.
\end{equation*}

Although these two classes of entanglement measures are defined in different ways, we
are able to make comparisons between them, and obtain the relations in
Proposition~\ref{propositions:relation-Esq-Er}.

\begin{proposition}
  \label{propositions:relation-Esq-Er}
  The following universal relations between entanglement measures hold:
  \begin{enumerate}
    \setlength\itemindent{6.25cm}
    \item $2E_{sq,c}^\infty \geq  E_r^\infty$,
    \item $2E_I             \geq  E_{r,\leftrightarrow}^\infty$,
    \item $2E_{sq}          \geq  E_{r,\rar}^\infty$.
  \end{enumerate}
  These relations also hold true if we replace the regularized entanglement measures
  by their corresponding non-regularized versions.
\end{proposition}

\begin{proof}
The first inequality is easy. For any classical extension
$\rho_{ABE}=\sum_i p_i\rho^i_{AB} \ox \proj{i}_E$ of a state $\rho_{AB}$, we have
\begin{equation*}\begin{split}
I(A;B|E)_\rho &=    \sum_i p_i D(\rho_{AB}^i \| \rho_A^i \ox \rho_B^i)  \\
              &\geq D\left( \rho_{AB} \| \sum_i p_i\rho_A^i \ox \rho_B^i \right)    \\
              &\geq E_r(\rho_{AB}),
\end{split}\end{equation*}
using the joint convexity of the relative entropy.
This, together with the definition of $E_{sq,c}$, implies that $2E_{sq,c}\geq E_r$.
Regularizing both sides, we get the regularized version as desired.

For the second inequality, we employ the idea for the proof of~\cite[Lemma 1]{BCY10}, and
apply it to the partial state merging protocol~\cite{YHW08}. Let $\rho_{AA'BB'E}$ be a
pure state, where the $AA'$ system is with Alice, $BB'$ is with Bob, and $E$ is at Eve's
hand. Alice and Bob are to transmit their systems $A$ and $B$ to Eve, by sending as less
as possible qubits to her, provided that unlimited entanglement is available between Alice
(Bob) and Eve. In the i.i.d. case, this task can be expressed as the transformation
$\rho_{AA':BB':E}^{\ox n}\longrightarrow\rho_{A':B':EAB}^{\ox n}$. Asymptotically, it
requires a minimal sum-rate $\frac{1}{2}\{I(AA';BB')-I(A';B')\}$ of quantum
communication~\cite{YHW08}. On the other hand, because the relative entropy of entanglement
$E_r$ is unlockable~\cite{Hors-Opp052},  the decrease of entanglement between Alice and Bob
in this protocol, measured by $E_r$, is no larger than $2$ times the qubits transmitted.
This means
\begin{equation}
  \label{eq:EI-Er-1}
  I(AA';BB')-I(A';B') \geq E_r^\infty(\rho_{AA':BB'})-E_r^\infty(\rho_{A':B'}).
\end{equation}
The right side of Eq.~(\ref{eq:EI-Er-1}) satisfies~\cite{Piani09}
\begin{equation}
  \label{eq:EI-Er-2}
  E_r^\infty(\rho_{AA':BB'})-E_r^\infty(\rho_{A':B'})\geq E_{r,\leftrightarrow}^\infty
                                (\rho_{AB})\geq E_{r,\leftrightarrow}(\rho_{AB}).
\end{equation}
Eqs.~(\ref{eq:EI-Er-1}) and (\ref{eq:EI-Er-2}), together with the definition of $E_I$,
lead to the second inequality and its non-regularized version as advertised.

The last inequality and its non-regularized version is essentially due to
Theorem~\ref{theorem:squshed-new-bound}, since squashed entanglement is non-increasing
under any LOCC operations.
\end{proof}

\begin{proposition}
  \label{proposition:relation-Errar-Ed}
  We have $E_{r,\rar}^\infty \geq  E_d$.
\end{proposition}

\begin{proof}
Let $\Lambda_n$ be a LOCC operation that satisfies
\[\|\Lambda_n(\rho_{AB}^{\ox n})-\Phi_{d_n}\|_1 \leq \epsilon\]
with $\epsilon\leq \frac{1}{e}$. We have
\begin{equation}
  \label{eq:Errar-Ed-2}
  E_{r,\rar}(\rho_{AB}^{\ox n})
     \geq E_{r,\ONELOCC}(\Lambda_n(\rho_{AB}^{\ox n}))
     \geq \log \frac{d_n+1}{2}-2\epsilon\log\frac{6d_n^2}{\epsilon},
\end{equation}
where the first inequality is by definition of $E_{r,\rar}$, and the second one makes use of
 Proposition~\ref{proposition:asycontinuity} and Proposition~\ref{proposition:EPR-evaluation}.
Recall that the distillable entanglement can be written as
\begin{equation}
  \label{eq:Errar-Ed-3}
  E_d(\rho_{AB})= \lim_{\epsilon\rar 0}\lim_{n\rar\infty}\sup_{\Lambda_n\in\LOCC}
  \left\{\frac{\log d_n}{n}:\|\Lambda_n(\rho^{\ox n})-\Phi_{d_n}\|_1\leq\epsilon\right\}.
\end{equation}
Eq.~(\ref{eq:Errar-Ed-2}) and Eq.~(\ref{eq:Errar-Ed-3}) together imply
\begin{equation*}\begin{split}
  E_d(\rho_{AB}) &\leq \lim_{\epsilon\rar 0}\lim_{n\rar\infty}
           \frac{E_{r,\rar}(\rho_{AB}^{\ox n})+1-2\epsilon\log\epsilon+2\epsilon\log 6}{n(1-4\epsilon)}\\
                 &= E_{r,\rar}^\infty(\rho_{AB})
\end{split}\end{equation*}
and we are done.
\end{proof}

\medskip
We summarize the relations between these entanglement measures in Fig.~\ref{fig1}. Since
we are mainly interested in the regularized versions, some relations between the
non-regularized entanglement measures are not reflected here. These include $E_{sq,c}
\leq E_f$, $E_r\leq E_f$, $E_r\leq 2E_{sq,c}$ and  $E_{r,\rar}\leq E_{r,\leftrightarrow}$ ($E_f$ is the
entanglement of formation). Some pairs of these entanglement measures are incomparable, meaning
that -- depending on the state -- they can be larger than each other. This is really the case
for $E_{sq}$ and $E_r$ ($E_r^\infty$). $E_{sq}\gg E_r$ is known for certain ``flower
states'', due to the lockability of $E_{sq}$ and non-locking of $E_r$~\cite{Hors-Opp052,
Christandl-Winter05}; the other direction $E_r^\infty\gg E_{sq}$ holds for $d\times d$
antisymmetric states~\cite{CSW09-1,CSW09-2}. We conjecture that the same situation occurs
between $E_I$ and  $E_r$ ($E_r^\infty$), $E_{sq}$ and $E_{r,\leftrightarrow}^\infty$
($E_{r,\leftrightarrow}$), $K_d$ and $E_{r,\leftrightarrow}^\infty$, $K_d$ and
$E_{r,\rar}^\infty$, which are left as open questions. Note that the possibility of
$E_I > E_r$ and $E_{sq} > E_{r,\leftrightarrow}^\infty$ for certain states, are known from
the relations in Fig.~\ref{fig1} and that $E_{sq}$ can be larger than $E_r$.

\begin{figure}[ht]
  \includegraphics[width=9.5cm]{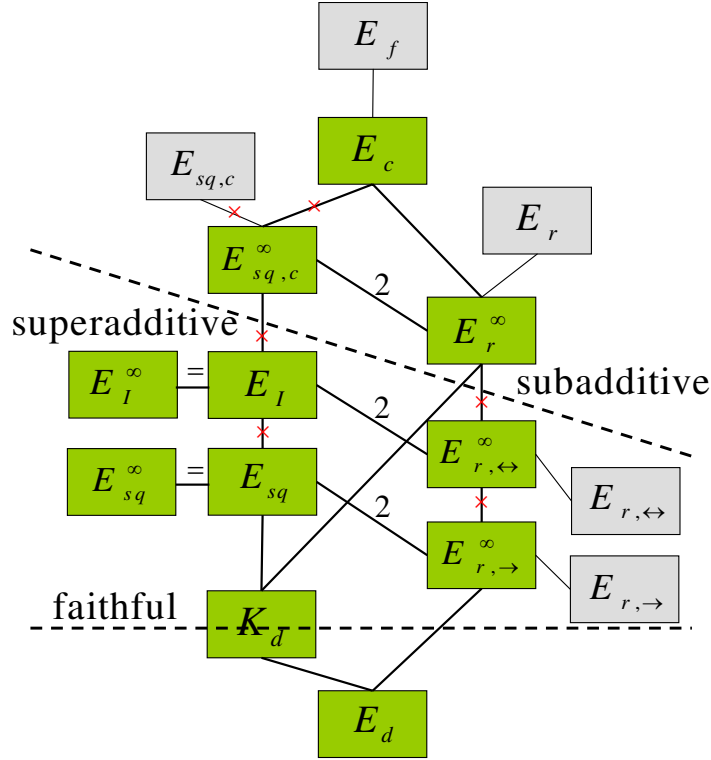}
  \caption{Relations between some entanglement measures. When two quantities are connected
  by a line with a constant above (constant $1$ is omitted), it means that the higher one
  multiplied by the constant is no smaller than the lower one. For those entanglement
  measures of which the separation is still unknown, we mark a red cross on the line that
  connects them. The upper dashed line divides
  these entanglement measures into two groups: the upper ones are subadditive and the lower
  ones are superadditive. Entanglement measures above the lower dashed line are faithful,
  while the only one below this line, $E_d$, is not faithful~\cite{Horodeckis98}. Whether
  the distillable key, $K_d$, is faithful or not, is still an open question. Hence, we put
  the line on it.}
  \label{fig1}
\end{figure}

\bigskip
The separation between entanglement measures is another interesting topic.
Proposition~\ref{proposition:EPR-evaluation} provides us with the strict inequalities
$E_{r,\leftrightarrow} < E_{r,\leftrightarrow}^\infty$ and $E_{r,\rar} < E_{r,\rar}^\infty$
for maximally entangled states. The fact that $2E_{sq,c}^\infty \geq  E_r^\infty$
(cf. Proposition~\ref{propositions:relation-Esq-Er}) and $E_r^\infty$ can be much larger than
$E_{sq}$ implies the separation between $E_{sq,c}^\infty$ and $E_{sq}$, disproving the
conjecture that $E_{sq,c}$ and $E_{sq}$ may be the same~\cite{YHW07}. Similarly, the relations
shown in Fig.~\ref{fig1}, together with the fact that $E_{sq}$ and $E_r^\infty$ can be much
larger than the other, lead to separations for the pairs $(E_c, E_r^\infty)$,
$(E_{sq,c}^\infty, E_r^\infty)$, $(E_r^\infty, E_{r,\rar}^\infty)$,  $(E_r^\infty, K_d)$,
$(E_I, E_{r,\leftrightarrow}^\infty)$, $(E_{sq}, E_{r,\rar}^\infty)$ and $(E_{sq}, K_d)$.
Separation between $E_{r,\rar}^\infty$ and $E_d$ is witnessed by the bound entangled states,
since the former is faithful. A separation between
$E_d$ and $K_d$~\cite{Hors-Opp05} had been discovered previously,
that between $E_c$ and $E_f$ is by Hastings~\cite{Hastings08, Shor03},
that between $E_r$ and $E_r^\infty$ due to Vollbrecht and Werner~\cite{Voll-Werner01}.

At last, separations between pairs of entanglement
measures that are still unknown, are marked in Fig.~\ref{fig1},
and we leave them as open questions.

\acknowledgments
We thank Fernando Brand\~{a}o, Matthias Christandl, Runyao Duan, Aram Harrow,
Masahito Hayashi and Dong Yang for helpful discussions.
AW was supported by the European Commission (STREP ``QCS'' and
IP ``QESSENCE''), the ERC (Advanced Grant ``IRQUAT''),
a Royal Society Wolfson Merit Award and a Philip Leverhulme Prize.
The Centre for Quantum Technologies is funded by the Singapore
Ministry of Education and the National Research Foundation as part
of the Research Centres of Excellence programme.


\end{document}